\documentclass[reqno]{amsart} \usepackage{amscd}
\usepackage{epsf}
\newtheorem{theorem}{Theorem}[section]
\newtheorem{proposition}[theorem]{Proposition}

\theoremstyle{remark} \newtheorem{remark}[theorem]{Remark}

\begin{document}

\title[Berezin and Odzijewicz quantizations on compact smooth manifolds]{   RKHS,  Berezin and Odzijewicz-type  quantizations on arbitrary  compact smooth manifold }

\author{ Rukmini Dey}
\address{Rukmini Dey,  email:rukmini@icts.res.in}

\maketitle

\begin{center}
\text{I.C.T.S.-T.I.F.R., Bangalore, India.}

\end{center}

\vskip .5in

Key words: Berezin quantization,  coherent states,  reproducing kernel, Odzijewicz quantization, deformation quantization.

\vskip .5mm

MSC:  81-XX,  53-XX,  43-XX, 46-XX, 47-XX.

\begin{abstract}

In this paper we define Berezin-type and Odzijewicz-type quantizations on compact smooth manifolds.  The method is as follows: we embed the smooth manifold of real dimension $n$ into ${\mathbb C}P^n$ and induce the quantizations from there.   The standard way by which reproducing kernel Hilbert spaces are defined on submanifolds gives a way to define the pullback coherent states.  In Berezin-type quantization the Hilbert space of  quantization  is the pullback (by the embedding) of the Hilbert space of geometric quantization of ${\mathbb C}P^n$.   In  the Odzijewicz-type quantization one has to consider a tensor product  of the geometric quantization line bundle with holomorphic $n$-forms.   In the Berezin case,  the operators that are quantized are those induced from the ambient space ${\mathbb C}P^n$.  The Berezin-type quantization exhibited here  is a generalization of an earlier  work  of the author and Ghosh.  In both Berezin and  Odzijewicz-type quantizations we first exhibit this quantization explicitly  on ${\mathbb C}P^n$ and  we induce the quantization on the smooth compact embedded manifold from ${\mathbb C}P^n$.
\end{abstract}

\section{\bf{Introduction}}

Some quantum systems donot come from quantizing classical systems (which are expected to have a symplectic structure) but there is a semi-classical limit of the quantum system, spin being such a system (\cite{Rad}).  We wish to include systems which donot have symplectic structure (or group action) and study if they are a semi-classical limit of some quantum system as $\hbar$ goes to zero. 
Many others have studied this "inverse problem'',  see for instance \cite{GoSh}.

The other motivation of the work is that  sometimes the Hilbert space of the problem turns out to be different from what the actual parameter space should prescribe.  The Hilbert space could be just be  induced  from geometric quantization of ${\mathbb C}^n$, or $ {\mathbb C}P^n$.  Roughly speaking in these two cases  the Hilbert space correspond to  multinomials.  In  some situations  this could be  at least a  good approximation,  for example the Quantum Hall Effect (where polynomials suffice for lowest Landau levels \cite{To}).  Quatum scars and fragmentation is another area in which this has direct application.  The dynamics of a quantum system could be restricted on a submanifold of the full Hilbert space of states in the Bose-Hubbard model,  for instance. 

  We wish to define two types of quantization,  namely the Berezin-type  and Odzijewicz-type quantizations on arbitrary compact smooth manifolds.  In our method we embed the arbitrary smooth manifold in question into a manifold which has an appropriate  ``quantization'',   a Hilbert space corresponding to the  the ``quantization'' and a Poisson structure.  A ``local'' Poisson structure is needed on the arbitrary compact smooth manifold and  it is induced from the  Poisson structure of  the manifold in which it is  embedded,  which could be   ${\mathbb C}P^n$ or ${\mathbb C}^n$ or some other suitable manifold.  The quantization is induced from ${\mathbb C}P^n$ or ${\mathbb C}^n$ (depending on whether we expect a finite or an infinite dimensional Hilbert space). 
In this paper we will concentrate on embeddings into ${\mathbb C}P^n$.

Both Berezin and Odzijewicz quantizations use coherent states in a very essential way. 
The  literature on coherent states  is vast,  see for instance a review by Ali,  Gazeau,  Antoine and Mueller \cite{AAGM},  Gazeau and Monceau  \cite{GM},    proceedings  by editors  Strasburger,  Ali,  Antoine,  Gazeau,  Odzijewicz \cite{AAGOS},  Perelomov \cite{Pe}.

In \cite{Be} Berezin  gave  a way of defining a star product on the symbol  of bounded linear operators acting on a Hilbert space (with a reproducing kernel) on a K$\ddot{\rm{a}}$hler manifold  (with certain conditions).  There is a parameter in the theory (namely $\hbar$)  such  that in the limit $\hbar \rightarrow 0$ the star product tends to usual product and the commutator  of the star product is proportional to the Poisson bracket upto first order.   This is called the correspondence principle.  After Berezin's orginal work \cite{Be},  Berezin quantization  has been generalized to many domains and manifolds, see Englis  \cite{E} for an example.  

We want to extend Berezin quantization to compact smooth manifolds.  We embed a compact smooth manifold into ${\mathbb C}P^n$ and pull back the reproducing kernel Hilbert space.  Pullback coherent states give symbols of bounded linear operators induced from those  corresponding to  ${\mathbb C}P^n$ and it is easy  to see that they  satisfy the correspondence principle. 

In  this context we recall that in \cite{DeGh} Dey and Ghosh  had considered pull back coherent states on  totally real submanifolds of ${\mathbb C}P^n$ and defined pull back operators and their ${\mathbb C}P^n$-symbols and showed that they satisfied the correspondence principle.   In other words we had shown Berezin-type quantization of certain operators for manifolds   which can be  embedded in ${\mathbb C}P^n$ as totally real submanifolds. This was part of Ghosh's thesis \cite{Gh} too.  Our present  work is a generalization of this as we donot need the condition totally real submanifolds.  The condition of totally real submanifolds  can give topological obstructions.  
We have no conditions on the compact smooth manifold.

Dey and Ghosh had  also   defined in  \cite{DeGh2}  a Berezin-type quantization on arbitrary even dimensional compact manifolds  (of real dimension $2n$) by removing a set of measure zero and embedding it into ${\mathbb C}P^n$.  This  was part of Ghosh's thesis \cite{Gh} too.  In this article our approach is slightly different.  We embed a $n$ real dimensional manifold into ${\mathbb C}P^n$.  

In \cite{Od} Odzijewicz studied the  quantization  of K$\ddot{\rm{a}}$hler manifolds (with some conditions)  using coherent state embedding and  showed that   the calculation of the Feynman path integral is  equivalent to finding the reproducing kernel function.  With the aim of extending Odzijewicz-type quantization on an  arbitrary compact smooth manfold,  we have to first exhibit this quantization  on ${\mathbb C}P^n$ explictly.  In the process we solve the Monge-Ampere equation on ${\mathbb C}P^n$.  The Monge-Ampere equation on compact K$\ddot{\rm{a}}$hler manifolds is known,  see for example \cite{Ko}.  We give an explict solution for ${\mathbb C}P^n$.  We generalize the Odzijewicz quantization to  arbitrary  smooth compact manifolds,  as we can  talk of reproducing kernel Hilbert spaces pulled back from ${\mathbb C}P^n$. 
We embed a compact smooth manifold into ${\mathbb C}P^n$ (using Whitney embedding or any other embedding) and pull back coherent states from ${\mathbb C}P^n$ \cite{DeGh}. 
Role of ${\mathbb C}P^n$ or ${\mathbb C}^n$  can  taken by  other appropriate manifolds too. 

All the  quantizations depend on the embedding.

It would be interesting to extend this method to Fedosov quantization \cite{Fe},   to works of Klauder \cite{Ka},  Klauder and Onofri \cite{Ka2} and also to  other examples mentioned in Odzijewicz's work \cite{Od} and the work of Brody and Graefe in \cite{BG}.  All this is work in progress.  We  are also working towards  applying  this method  to quantum hall effect in higher dimensions,  see  works of Karabali and Nair  \cite{KarNa} and Karabali \cite{Kar} for ${\mathbb C}P^n$.   We are applying our theory to the works of  H. Omori, Y. Maeda and A. Yoshioka \cite{OMY}. We also are studying quantum scars in this context,  see for instance the work of W. Ho,  S. Choi,  H.  Pichler,  M.  Lukin \cite{HCPL}.

\section{{\bf Review of Geometric  Quantization and Coherent states on ${\mathbb C}P^n $}}

 This review follows \cite{Be} and \cite{Gh} closely.

Let  $U_0 \subset {\mathbb C}P^n$ given by $U_0 = \{\mu_0 \neq 0\}$ where $[\mu_0,...., \mu_n]$ are homogeneous  coordinates on ${\mathbb C}P^n$.   
Let $(\mu_1, \mu_2,...\mu_n)$ be inhomogenous coordinates on $U_0 \equiv {\mathbb C}^n$ such that  $[1, \mu_1, \mu_2, ...,\mu_n]  \in U_0$.

The Fubini-Study form is given by $\Omega_{FS} = \sum_{i, j=1}^n \Omega^{FS}_{ij} d \mu_i \wedge d \bar{\mu}_j, $ where  the K$\ddot{\rm{a}}$hler metric $G$ and the K$\ddot{\rm{a}}$hler form $\Omega_{FS}$ are related by
$\Omega_{FS} (X, Y) = G(I X, Y)$. 

The Poisson bracket of two functions $t$ and $s$:

$\{t,s\}_{FS} = \sum_{i,j=1}^n \Omega_{FS}^{ij} \left(\frac{\partial t}{\partial \bar{\mu}_j}  \frac{\partial s}{\partial \mu_i}  - \frac{\partial s}{\partial \bar{\mu}_i}  \frac{\partial t}{\partial \mu_j} \right)$ where $(\Omega_{FS}^{ij})$ are the matrix coefficients of the inverse of the matrix $(\Omega^{FS}_{ij})$ of the Fubini-Study form.

 Let $H^{\otimes m}$ be the $m$-th tensor product of the hyperplane section bundle $H$ on $CP^n$.  Recall that $m \Omega_{FS}$ is its curvature form and $m \Phi_{FS}$ is a local K$\ddot{\rm{a}}$hler potential where $e^{m \Phi_{FS} (\mu,\bar{\mu})} = (1 + |\mu|^2)^m$. 
Let  $\{\phi_i \}_{i=1}^N$ be an orthonormal basis for the space of holomorphic sections.

On $U_0$ the sections of $H^{\otimes m}$ are functions since the bundle is trivial when restricted to $U_0$. They can be identified with polynomials in $\{\mu_i\}_{i=1}^n$ of degree at most $m$.

Let $\hbar = \frac{1}{m}$ be a  parameter. Then $\{\phi_i\}$  depend on $\hbar$. 

We define $dV(\mu)  =  {\mathcal G}(\mu)  \Pi_{i=1}^{n} |d \mu_i \wedge d \bar{\mu}_i| = {\mathcal G}(\mu)  |d \mu \wedge d \bar{\mu}| = \frac{|d \mu \wedge d \bar{\mu}|}{(1 + |\mu|^2)^{n+1}} $ to be a volume form on $ {\mathbb C}P^n$,  restricted to $U_0$ and where ${\mathcal G} = \det[g_{FS}^{ij}|_{U_0}]$.  
 
Then $ V=\int_{{\mathbb C}^n} dV = \int_{{\mathbb C}^n} \frac{|d \mu \wedge d \bar{\mu}|}{(1 + |\mu|^2)^{n+1} } < \infty$.

Let  $(c(m))^{-1}  = \int_{U_0} \frac{1} { (1+  |\nu|^2)^m} dV(\nu)= \int_{U_0} e^{-m \Phi_{FS} (\nu,\bar{\nu})} dV(\nu).$

Let an innerproduct   on the space of functions on $U_0$ be defined as 
  
  $$\left<f,g \right> = c(m) \int_{U_0} \frac{ \overline{f(\nu)}g(\nu)}{ (1 + |\nu|^2)^m}  dV(\nu) =  c(m) \int_{U_0} \overline{f(\nu)}g(\nu) e^{-m \Phi_{FS}(\nu,\bar{\nu})}   dV(\nu).$$

 Also,   
$D_{(q_1, q_2,...q_n;q)} = c(m) \int_{U_0} \frac{|\nu_1|^{2q_1}...|\nu_n|^{2q_n}} {  (1+  |\nu|^2)^m} dV(\nu) ,$ where $q_i's$ are all possible positive  integers such that   $q_1+ ...+q_n = q;q=0,...,m.$ 

Let $$\Phi_{(q_1,q_2,...,q_n;q)}(\mu) = \frac{1}{\sqrt{D_{(q_1,...,q_n;q)}}} \mu_1^{q_1}...\mu_n^{q_n},$$ where $q_1+ ...+q_n = q;q=0,...,m.$

It is easy to check that 
 $\{ \Phi_{(q_1,...,q_n; q)} \}$ are   orthonormal in ${\mathbb C}^n$ with respect to the inner product defined as above  and are restriction of  a basis for sections of $H^{\otimes m}$  to $U_0$.  The span of these form a Hilbert space with the above inner product.

{\bf Definition}
The Rawnsley-type coherent states \cite{Ra}, \cite{Sp}  are given on $U_0$  by $\psi_{\mu}$ reading as follows:

$\psi_{\mu} (\nu) := \sum_{{q_1+q_2+...+q_n=q;q=0,1,...,m}} \overline{\Phi_{(q_1,q_2,...,q_n;q)}(\mu)}\Phi_{(q_1,q_2,...,q_n;q)}(\nu).$

In short hand notation 
$\psi_{\mu}  :=   \sum_{{I}} \overline{\Phi_{I}(\mu)}\Phi_{I}$ where  the multi-index $I = (q_1,...,q_n;q)$ runs over the set 
  $q_1 + ...+ q_n = q; q=0,...,m$.

This is a reproducing kernel in the sense below.

\begin{proposition}\label{kernel}
Reproducing kernel property.
If  $\Psi$ is any other section,  then $\left< \psi_{\mu}, \Psi \right> =   \Psi(\mu)$.  In particular, $\left< \psi_{\mu}, \psi_{\nu} \right> =   \psi_{\nu}(\mu)$.
\end{proposition}

\begin{proposition}\label{resolution}
Resolution of identity property:
$$c(m) \int_{U_0} \left< \Psi_1 , \psi_{\mu} \right> \left< \psi_{\mu} , \Psi_2 \right>  e^{-m \Phi_{FS}(\mu,\bar{\mu})} dV(\mu) = \left< \Psi_1 ,  \Psi_2 \right>.$$ In particular, 

$$c(m) \int_{U_0} \left< \psi_{\nu},   \psi_{\mu} \right>\left< \psi_{\mu},  \psi_{\nu} \right>  e^{-m \Phi_{FS}(\mu,\bar{\mu})} dV (\mu)= \left< \psi_{\nu} ,  \psi_{\nu}  \right>. $$
\end{proposition}

The proofs of these are easy and can be found for instance  in   \cite{DeGh2}.  It is in general true of Rawnsley type coherent states.

\section{{\bf A reproducing kernel Hilbert space on a compact smooth manifolds and coherent states }}

In this section we construct a  reproducing kernel Hilbert space and coherent states  on a compact smooth manifold of real dimension $n$  by emedding it into ${\mathbb C}P^n$.  This generalizes a result in \cite{DeGh}.
 We proceed similar to  \cite{DeGh},  but we donot need the ``totally real''  condition.  We use the Whitney embedding of any compact smooth manifold.  Any other smooth embedding will also work.

Let $X^n$ be a compact smooth manifold of real dimension $n$.   Let $\chi: X \rightarrow {\mathbb R}^{2n}$ be any embedding (for instance Whitney embedding).  Let $ i : {\mathbb R}^{2n} \rightarrow  {\mathbb C}P^n$ be the inclusion such that $R^{2n} $ is identified with $U_0 \subset {\mathbb C}P^n$ and $\epsilon = i \circ \chi$.
 It is clear that $\epsilon : X \rightarrow {\mathbb C}P^n$ is an embedding and 
that $\epsilon: X  \rightarrow \epsilon(X)$ is a diffeomorphism.  Let $\Sigma = \epsilon (X)$. 

 Let ${\mathcal H}_m$ be the sections of $H^{\otimes m}$ with norm denoted for short as $\lvert \lvert s \rvert \rvert_{{\mathbb C}P^n}$.  
 
 Let $\psi_{(q_1,q_2,...,q_n;q)}(\mu) = \frac{1}{\sqrt{D_{(q_1,...,q_n;q)}}} \mu_1^{q_1}...\mu_n^{q_n}$ where $q_1+ ...+q_n = q;q=0,...,m$  be an orthonormal basis for it as mentioned in the previous section.

 Let ${\mathcal H}_{2m} = \epsilon^*({\mathcal H}_m) $ be the pullback Hilbert space on $X$.  Thus if $\tilde{s} \in  {\mathcal H}_{2m}$, it is of the form $\tilde{s}= s  \circ \epsilon$.  The norm on ${\mathcal H}_{2m}$ is given by 
$$\lvert  \lvert \tilde{s}  \rvert \rvert_X = \rm{min}_{s \in {\mathcal H}_m} \{ \lvert \lvert  s   \rvert \rvert_{{\mathbb C}P^n}  :  \tilde{s}= s  \circ \epsilon \}. $$

\begin{proposition}
  $  {\mathcal H}_{2m}$  is a   Hilbert space in the  $ \lvert \lvert \cdot   \rvert \rvert_X$ norm.
\end{proposition}

\begin{proof}
For proof see for instance  \cite{PaRa}.
\end{proof}

There are two ways of defining reproducing kernel (coherent states).  For the first one we follow \cite{DeGh}. 

Let $\eta_I$ be an orthonormal basis for ${\mathcal H}_{2m}$,  with the  norm $\lvert  \lvert \cdot \rvert \rvert_X$.

{\bf Definition}
The Rawnsley-type coherent states on $X$ are defined locally as 
\begin{equation}\label{cohX}
\Phi_{p} = \sum_{k=1}^l \overline{\eta_{I_k} (p)} \eta_{I_k}{\rm \;  where \;} p \in X. 
\end{equation}

For a global defintion,  see \cite{DeGh}.
As before they are overcomplete,  have reproducing kernel property,  resolution of identity property.

 The second method is pulling back the kernel on ${\mathbb C}P^n$ to  $X$.  We will use this method henceforth.
 
 {\bf Definition:}
 For $p \in X$,  let $\Psi_p$ be defined as
 \begin{equation}\label{pbcoh}
 \Psi_p(\cdot)= \epsilon^*\psi_{\epsilon(p)}(\cdot)  := \epsilon^* \Big{(} \sum_{{q_1+q_2+...+q_n=q;q=0,1,...,m}} \overline{\Phi_{(q_1,q_2,...,q_n;q)}(\epsilon(p))}\Phi_{(q_1,q_2,...,q_n;q)}(\cdot) \Big{)}.
 \end{equation}

 This is the same as the pull back of the coherent states $\psi_{\mu}$ on ${\mathbb C}P^n$ where $\mu = \epsilon(p)$, i.e.  the pullback of $\psi_{\epsilon(p)}$ to $X$,  namely $\epsilon^*(\psi_{\epsilon(p)} )$.
  By \cite{PaRa} (Prop. (5.6) and Th. (5.7)) we have $\Psi(p,q) := \Psi_p(q) = \psi_{\epsilon(p)} (\epsilon(q)) = \psi(\overline{\epsilon(p)},  \epsilon(q)) $ is a reproducing kernel on ${\mathcal H}_{2m}$. 
  
We will need to use this fact crucially in induced  quantizations.
 
 \section{{\bf Induced Berezin-type quantization on compact  smooth manifolds}}

Let $X$ be a compact smooth manifold. 
Let $\epsilon: M \mapsto {\mathbb C}P^n$ as in section 3. 
Let us continue on ${\mathbb C}P^n$ and recall the Berezin quantization on it.

\subsection{ Review of Berezin quantizaion on ${\mathbb C}P^n$:}

Let $\psi_{\mu}$ be defined as in section 2.   As in \cite{Be},  we denote by 
 ${\mathcal L}_{m} (\mu,  \bar{\mu}) = \left< \psi_{\mu}, \psi_{\mu} \right> = \psi_{\mu}(\mu),$ 
 ${\mathcal L}_{m} ( \mu,  \bar{\nu}) =\left< \psi_{\mu}, \psi_{\nu} \right> = \psi_{\nu}(\mu).$

Let $\hat{A}$ be a bounded linear operator acting on ${\mathcal H}$.  Then,  as in  \cite{Be},  one can define a symbol of the operator as
$$A(\nu, \bar{\mu}) = \frac{\langle  \psi_{\nu} , \hat{A} \psi_{\mu} \rangle}{ \langle  \psi_{\nu},  \psi_{\mu} \rangle}.$$

One can show that one can recover the operator from the symbol by a  formula \cite{Be}.

Let $\hat{A}_1,  \hat{A}_2$ be two such operators and let $\hat{A_1} \circ \hat{A_2}$ be their composition.

Then the symbol of $\hat{A_1} \circ \hat{A_2}$ will be given by the star product defined as in \cite{Be}:
\begin{eqnarray*}\label{star}
& & (A_1 * A_2 ) (\mu, \bar{\mu}) \\
&=&  c(m) \int_{U_0} A_1(\mu,  \bar{\nu}) A_2 (\nu,  \bar{\mu}) \frac{{\mathcal L}_{m} (\mu, \bar{\nu}) {\mathcal L}_{m} (\nu, \bar{\mu})}{ {\mathcal L}_{m} (\mu, \bar{\mu}){\mathcal L}_{m} (\nu, \bar{\nu})}  {\mathcal L}_m(\nu, \bar{\nu}) e^{-m\tilde{\Phi}(\nu, \bar{\nu})} d V(\nu),
\end{eqnarray*}
where recall $\frac{1}{c(m)} = \int_{U_0}  e^{-m \Phi_{FS} (\nu, \bar{\nu})} dV(\nu). $

This is the symbol of $\hat{A_1} \circ \hat{A_2}$.

\medskip

One can show   the following formula gives  the reproducing kernel \cite{DeGh2},  \cite{Gh}.
\begin{proposition}\label{multi}
$\psi_{\mu}(\nu) =  ( 1 + \bar{\mu} \cdot \nu)^m$.
\end{proposition}

\medskip

\begin{theorem} [Berezin]\label{corresp}

Let $\mu \in  {\mathbb C}^n$.  

The star product satisfies the correspondence principle:

1. $ \lim_{m \rightarrow \infty} (A_1 \star A_2)(\mu, \bar{\mu}) = A_1(\mu, \bar{\mu}) A_2 (\mu, \bar{\mu}),$

2. $ \lim_{m \rightarrow \infty} m (A_1 \star A_2 - A_2 \star A_1)(\mu, \bar{\mu}) =  i \{ A_1, A_2\}_{FS} (\mu, \bar{\mu}).$
\end{theorem}
See \cite{Be}, \cite{Gh} for proof.

\subsection{Induced operators on $X$  and correspondence principle}
  
 We turn to $X$.
Let $\hat{B}$ be a bounded linear operator from ${\mathcal H}_{2m}= \epsilon^*({\mathcal H}_m)$ to itself.  Let $\hat{A}$  be defined by  the least  norm operator such that 
\begin{equation}\label{AB}
\hat{B} (\epsilon^*(s)) :=  \epsilon^* (\hat{A} s)
\end{equation}
i.e.  if $\hat{B} =  \epsilon^*(\hat{A})= \epsilon^*(\hat{A}_1)$, then $\lvert \lvert \hat{A}_1   \rvert \rvert \geq  \lvert \lvert \hat{A}   \rvert \rvert $.  Here $\hat{A}$ is  a bounded linear operator from ${\mathcal H}_m$ to itself.

Let $\Psi_p$ be the coherent states defined in the previous section,  namely the pullback coherent states obtained from pulling back the reproducing kernel. 

{\bf Definition} Let $B: X \times X \rightarrow {\mathbb C}$ be the symbol of $\hat{B}$ in the coherent states $\Psi_p$, i.e. 
$B(p,p)  = \frac{\langle \Psi_p,  \hat{B}(\Psi_p) \rangle_X}{\langle \Psi_p, \Psi_p \rangle_X} $ and 
$B(p,q)  = \frac{\langle \Psi_p,  \hat{B}(\Psi_q) \rangle_X}{\langle \Psi_p, \Psi_q \rangle_X} .$ This is named as  the $X$-symbol of $\hat{B}$.  Here the norm  is defined as in the previous section.

The crucial point is $B(p,q) = \frac{\hat{B}(\Psi_q) (p)}{ \Psi_q(p)}$.

Let $ \Psi_p = \epsilon^*\psi_{\epsilon(p)}$ be defined as in  Eq.(\ref{pbcoh}).

Then the  definition of the symbol of $\hat{A}$ gives that  $$A(\epsilon(p), \epsilon(q)) =  \frac{\langle \psi_{\epsilon(p)} , \hat{A} \psi_{\epsilon(q)}\rangle_{\mathbb{C}P^n}}{ \langle  \psi_{\epsilon(p)},  \psi_{\epsilon(q)} \rangle_{\mathbb{C}P^n}}  =  \frac{(\hat{A} \psi_{\epsilon(q)})(\epsilon(p)) }{  \psi_{\epsilon(q)} (\epsilon(p))}.  $$

Then $\epsilon^*(\hat{A} \psi_{\epsilon(q)}) =  \hat{B} \Psi_q$ by Eq. (\ref{AB}) and definition of $\Psi_q$.

\begin{proposition}
  $B(p,q) = A(\epsilon(p),  \epsilon(q))$.
\end{proposition}

\begin{proof}

\begin{eqnarray*}
  A(\epsilon(p), \epsilon(q))   =  \frac{\epsilon^*\hat{A}( \psi_{\epsilon(q)})(p) }{  \epsilon^*\psi_{\epsilon(q)} (p)}=   \frac{\hat{B} (\Psi_q)(p)}{ \Psi_q(p)}.
\end{eqnarray*}
The last equality follows since $\epsilon^*( \hat{A} \psi_{\epsilon(q)}) = \hat{B} (\Psi_q)$. 
\end{proof}
Let $B_1$ and $B_2$ be the $X$-symbols of $\hat{B}_1$ and $\hat{B}_2$  bounded linear  operators on ${\mathcal H}_{2m}$.  Let
$\hat{A}_1 , \hat{A}_2$ be two bounded linear operators on ${\mathcal H}_{m}$ which are of least norm satisfying  Eq. \ref{AB}.   Let $A_1 * A_2$ be the symbol of 
$A_1 \circ A_2$.  The formula for this can be found in \cite{Be}.  Then  $B_1 * B_2 $ is the $X$-symbol of $\hat{B}_1 \circ \hat{B}_2 $.

We know  $A_1 $ and $A_2$ satisfy the correspondence principle,  by Th. (\ref{corresp}).

\begin{theorem}
The star product on the symbol of bounded linear operators $\hat{B}_1$ and $\hat{B}_2$ on ${\mathcal H}_{2m}$ satisfies the correspondence principle:

(1) $ \lim_{m \rightarrow \infty} (B_1 * B_1)(p,p) = B_1(p,p) B_2(p,p)$.

(2) $\lim_{m \rightarrow \infty} m ( B_1 * B_2 - B_2 * B_1)(p,p) = i \{B_1,  B_2\}_{FS}(p,p).$

\end{theorem} 

\begin{proof}
This follows from the fact that $A_1 * A_2$ satisfy the correspondence principle in ${\mathbb C}P^n$ and $B(p,q) = A(\epsilon(p),  \epsilon(q))$.
\end{proof}

 \section{{\bf Odzijewicz-type  quantization on compact smooth manifolds}}

\subsection{{\bf Review of Odzijewicz quantization on ${\mathbb C}P^n$}}

Let  ${\mathbb C} P^n= \cup_{\alpha=0}^n U_{\alpha}$, where $U_{\alpha}$ is the set of  $[(\mu_0,...\mu_{i-1}, 1, \mu_{i+1},...,\mu_n)]$, $U_\alpha$ are the inhomogenous coordinate neighborhoods.

(We will use $\alpha, \beta$ indices ( not $i,j$ indices ) to be consistent with Odzijewicz notation).

The Odzijewicz quantization goes through in the same spirit  as in his paper  \cite{Od} when one takes the K$\ddot{\rm{a}}$hler manifold $M$ to be ${\mathbb C}P^n$. 

Let ${\mathcal M }$ be the Hilbert space of holomorphic sections of $H^{\otimes m} \otimes T^{*(n,0)}( {\mathbb C}P^n)$ where $H$ is the hyperplane section bundle and $m$ is a large  positive integer.  The holomorphic sections are $H^{\otimes m}$-valued holomorphc $n$ forms.  If $s_1 = f d\mu_1 \wedge...\wedge d\mu_n $ and $s_2 = g d\mu_1 \wedge ...\wedge d \mu_n$ are two sections,  the innerproduct is given by
$$\langle s_1, s_2 \rangle _{\mathbb{C}P^n} = c(m) \int_{U_0} \frac{ \overline{f(\nu)}g(\nu)}{ (1 + |\nu|^2)^m}  dV(\nu)$$ where 
$dV(\mu)  =  {\mathcal G}(\mu)  \Pi_{i=1}^{n} |d \mu_i \wedge d \bar{\mu}_i| = {\mathcal G}(\mu)  |d \mu \wedge d \bar{\mu}| = \frac{|d \mu \wedge d \bar{\mu}|}{(1 + |\mu|^2)^{n+1}} $ to be a volume form on $ {\mathbb C}P^n$, restricted to $U_0$ and where ${\mathcal G} = \det[g_{FS}^{ij}|_{U_0}]$. 
  
In this case ${\mathcal M}$ is sufficiently ample and when restricted to  $U_0 \subset {\mathbb C}P^n$ the coherent state embedding is given by (notation as in \cite{Od})
$$\mu \in {U_0 \subset \mathbb C}P^n \mapsto   [K_0(\bar{\mu}, \cdot)] \in {\mathbb C}P({\mathcal M}).$$

We have for $v \in U_{\alpha}$ and $z  \in U_{\beta}$
 $$ K_{\alpha}(\bar{v},  z) = K_{\alpha \beta} (\bar{v},  z) s_{\beta} dz_\beta^1 \wedge ....\wedge d  z_\beta^n,$$
where $s_{\beta}$ is the unit section of $H^{\otimes m}$ on $U_{\beta}$ and $( z_{\beta}^1, ..., z_{\beta}^n) $ are coordinates on $U_{\beta}$.

Then $K_{\alpha \beta}$ satisfy $(2.13) - (2.16)$ in \cite{Od}.

We will sometimes use the notation $v,z$ and sometimes $\mu, \nu$,  latter when we are on $U_0$.

We can define the transition probability between $\mu,  \nu \in U_0$ to be 
$$ a_{00}(\bar{\mu}, \nu) = \frac{K_{00}(\bar{\mu},  \nu)}{K_{00} (\bar{\mu}, \mu) ^{1/2} K_{00}(\bar{\nu}, \nu)^{1/2}}.$$
where $a_{00}$ is the transition amplitude of points in $U_0$.

More generally, 
if  $v \in U_{\alpha}$ and $z\in U_{\beta}$, $  \alpha, \beta =0,1,2,...,n$, then the transition probability amplitude is given by

$$ a_{ \alpha \beta}(\bar{v}, z) = \frac{K_{ \alpha \beta}(\bar{v}, z)}{K_{\alpha \alpha } (\bar{v}, v) ^{1/2} K_{\beta \beta}(\bar{z}, z)^{1/2}}.$$

Odzijewicz \cite{Od} shows that the transition probabilty amplitude between two points  can be written as an integral which involves the coherent states and a measure using a partition of unity (see 2.21 in \cite{Od}) . However the natural measure is the Lebesgue measure and
we assume that the two measures differ by a positive constant.

 For this, we need  to solve the Monge-Ampere (equation $(2.22)$ in \cite{Od})  on $U_0 \subset {\mathbb C}P^n$. 
This is not a new result, see for example \cite{Ko},  but we give an explicit proof in the following. 
\begin{proposition}
The Monge-Ampere equation on ${\mathbb C}P^n$,  namely, 
\begin{equation}
\rm{det}[ \frac{\partial^2 \rm{log} \rho_{00}(\mu)}{\partial \mu_j \partial \bar{\mu}_k} ]= C (-1)^{\frac{n(n+1)}{2}} \frac{1}{n!} \rho_{00}(\mu)  K_{00}(\bar{\mu}, \mu)
\end{equation}
has a solution with  $C>0$ when $  (-1)^{\frac{n(n+1)}{2} +1}   $ is positive and $C<0$ when $(-1)^{\frac{n(n+1)}{2} +1}$ is negative.
\end{proposition}

\begin{proof}
We take $\rho_{00}(\mu) = \frac{1}{(1+ |\mu|^2)^{N}}  = \rm{exp}(- N (log(1 + |\mu|^2)))$.  We know that  $K_{00}(\bar{\mu}, \mu) =  \psi_{\mu}(\mu)= (1 + |\mu|^2)^m $. 
The left hand side of the Monge-Ampere equation is $\frac{-N}{(1 + |\mu|^2)^{n+1}}$.

$$\frac{-N}{(1 + |\mu|^2)^{n+1}}  = C(-1)^{n(n+1)/2} \frac{1}{n!}  \frac{1}{(1+ |\mu|^2 )^{N-m}}.$$

Then $N = n+m +1$ and $C = \pm N n!$ (positive or negative depending on whether $(-1)^{\frac{n(n+1)}{2} +1}$  is positive or negative).   
\end{proof}

One has a coherent state embedding which when restricted to $U_0$ is  $\mu  \in {\mathbb C}P^n \mapsto [ K_{0}(\bar{\mu},  \cdot)] \in {\mathbb C}P({\mathcal M})$.
One can define path integral as in $(2.28)$  \cite{Od} and show that it is related to the pull back metric on ${\mathbb C}P^n$ of the Fubini-Study metric on $ {\mathbb C}P({\mathcal M})$.  Odzijewicz defines the transtiion probability amplitudes as 
$$a_{\alpha \beta}(\bar{v}, z)  = \frac{K_{\alpha \beta}(\bar{v},  z)}{K_{\alpha \alpha}(\bar{v},  v)^{\frac{1}{2}} K_{ \beta \beta} (\bar{z}, z)^{\frac{1}{2}}},$$ where $v \in U_{\alpha},  z  \in U_{\beta}$,  $\alpha, \beta =0,...,n$.

Odzijewicz shows that given a path $\gamma$ joining $v \in U_{\alpha}$ and $z \in U_{\beta}$ is in fact given in the following proposition.  
\begin{proposition}[Odzijewicz, \cite{Od}]\label{holonomy}
$a_{\gamma, \alpha \beta}(\bar{v},  z) = \exp [i \int_{\gamma} \rm{Im} ( \bar{\partial} \log K)].$
\end{proposition}
The proof of this can be found in Odzijewicz \cite{Od2}.
We also give a simple proof of the same.

\begin{proof}
Let $\gamma$  be a curve beginning at $z$ and ending at $v$ and let it be parametrized by $\tau$. It can  be subdivided such that  $z_i = \gamma(\tau_i)$, $i =1,...,N$ s.t.  $z_1 = z $ and $z_N = v$. 
The transition probability amplitude from state  $z$ to state $v$ with necessray transtion through $z_i$ , $i =2,...,N-1$ is given by $\Pi_{i=1}^{N-1} a_{\alpha_{i +1} \alpha_i}(\bar{z}_{i+1}, z_i)$. 

It is simple to see that 

$$a_{\alpha \beta} (\gamma, \bar{v}, z) = \lim_{N\rightarrow \infty} \exp \sum_{i=1}^{N} \log a_{\alpha_{i+1} \alpha_i}(\bar{z}_{i+1}, z_i)$$ (notation as in \cite{Od}).
Let the subdivisions be so small that compared to $z_{i+1} - z_i = \Delta z_i  =h_i$, $(h_i)^2 $ is negligible.  We write $z_{i+1} = z_i + h_i$ and use Taylor series expansion and neglect $(h_i)^2$ terms compared to $h_i$.
\begin{eqnarray*} 
\log a_{\alpha_{i+1} \alpha_i}(\bar{z}_{i+1}, z_i) &=& \log K_{\alpha_{i+1} \alpha_i}(\bar{z}_{i+1}, z_i) -  \frac{1}{2} \log K_{\alpha_{i+1}\alpha_{i+1}} ( \bar{z}_{i+1}, z_{i+1})  \\&-& \frac{1}{2} \log K_{\alpha_i \alpha_i}( \bar{z}_i, z_i)\\
&=&  \frac{1}{2}  \bar{\partial} \log K_{\alpha_i \alpha_i}(\bar{z}_i, z_i) \bar{h}_i
-  \frac{1}{2}  \partial \log K_{\alpha_i \alpha_i} (\bar{z}_i ,z_i) h_i  \\
& & + \rm{higher \; order \; terms}\\
&=& i \rm{Im} \bar{\partial} \log K_{\alpha_i \alpha_i} (\bar{z}_i, z_i) \Delta z_i.
\end{eqnarray*}
Then,  using $$a_{\alpha \beta} (\gamma, \bar{v}, z) = \lim_{N\rightarrow \infty} \exp \sum_{i=1}^{N_1} \log a_{\alpha_{i+1} \alpha_i}(\bar{z}_{i+1}, z_i)$$ we get the result. 
\end{proof}

\begin{remark}
Odzijewicz  \cite{Od} further shows  that  $a_{\gamma, \alpha_i \beta_i} (\gamma, \bar{v} , z) $ can be interpreted as a parallel transport  along $\gamma$ w.r.t.  a  connection on the line bundle $H^{\otimes m} \otimes T^{*(n,0)}({\mathbb C}P^n)$. 

This leads him to define a path integral to be the transiltion probability amplitude   $a_{\alpha\beta} $.  In other words
$$a_{\alpha \beta}(\bar{v}, z) := \int {\mathcal D} [\gamma] \exp[i \int_{\gamma} \rm{Im} (\bar{\partial} \log K)].$$
Here the integral is over all paths joining $z$ and $v$.

This intepretation is in keeping with a certain  theory with zero Hamiltonian,  see Odzijewicz formalism in \cite{Od2}. See also Klauder's formalism in \cite{Ka}, \cite{Ka2} and Spera \cite{Sp}.
\end{remark}

n \cite{Od}  Odzijewicz shows  the following.  Let $M$ be the K$\ddot{\rm{a}}$hler manifold of his consideration, which for us is $\mathbb{C}P^n$.  $M$ is embedded in $\mathbb{C}P(\mathcal{M})$ by the coherent state embedding.
The standard metric on $\mathbb{C}P(\mathcal{M})$ is given by 
$$d([s_1], [s_2]) := \rm{inf}_{t_1, t_2}  \lvert \lvert \frac{e^{it_1} s_1}{||s_1||}, \frac{e^{it_2} s_2}{||s_2||} \rvert \rvert. $$

Putting $s_1 = K_{\bar{\alpha}}(\bar{z}_1, \cdot)$ and $s_2 = K_{\bar{\beta}}(\bar{z}_2, \cdot),$ $z_1 \in \Omega_{\alpha}$ and $z_2 \in \Omega_{\beta}$.  Let $d_M $ be the pullback metric of $d$ to $M$ (by the coherent state embedding ). Then, 
\begin{proposition}[Odzijewicz, \cite{Od}]
$$d_{M} (\mu, \nu) = \sqrt{2} ( 1 - | a_{ \alpha \beta}(\bar{\mu}, \nu)|)^{\frac{1}{2}}.$$ 
\end{proposition}

 We give a simple proof. 
 
 \begin{proof}
 We have to show that $d_M^2 = 2( 1 - | a_{ \alpha \beta}(\bar{\mu}, \nu)|)$.

 \begin{eqnarray*}
 d^2 &=& \rm{ inf}_{t_1, t_2} \lvert \lvert  \frac{e^{it_1} s_1}{||s_1|| } -  \frac{e^{it_2} s_2}{||s_2|| } \rvert \rvert^2 \\
 & =& \rm{inf}_{t_1, t_2} (2 - e^{i(t_2 - t_1)} \frac{\langle s_2, s_1 \rangle}{||s_2|| ||s_1||} -  e^{i(t_1 - t_2)} \frac{\langle s_1, s_2 \rangle}{||s_2|| ||s_1||}).
  \end{eqnarray*}
 Since $s_1 = K_{\bar{\alpha}}(\bar{z}_1, \cdot)$ and $s_2 = K_{\bar{\beta}}(\bar{z}_2, \cdot),$ by definition we can see that 
 $\frac{\langle s_1, s_2 \rangle}{||s_2|| ||s_1||} = a_{\alpha \beta} (\bar{z}_1, z_2)$.  Let $A= a_{\alpha \beta} (\bar{z}_1, z_2)$ where $A = |A| e^{i x}$. 
 Let $\theta = t_2 - t_1$. 
 
 Then 
 \begin{eqnarray*}
 d_M^2 &=& \rm{inf}_{t_1, t_2} (2 - 2 \cos (\theta)  \rm{Re} A - 2 sin (\theta) \rm{Im} A)  \\
 &=&  \rm{inf}_{t_1, t_2} (2 - 2 |A|  \cos( \theta -x )) = (2 - 2 |A|) \\
 &=& (2 - 2 | a_{\alpha \beta} (\bar{z}_1, z_2)|).
 \end{eqnarray*}
 Thus the result follows.  

\end{proof}

\subsection{Induced Odzijewicz-type quantization on compact smooth manifolds}

Let $X$ be a compact smooth manfiold of real dimension $n$.
Let $U_0$ be the open subset of ${\mathbb C}P^n$ as mentioned in section 2.  We know as a topological space $U_0 \equiv {\mathbb R}^{2n}$. 
Let $\epsilon : X \mapsto U_0 \subset {\mathbb C}P^n$ be an embedding (for example Whitney embedding). Let $\Sigma = \epsilon(X)$.

We change notations in this section to be consistent with Odzijewicz's notation in \cite{Od}. 

Let $p,q \in X$ and $v = \epsilon(p)$ and $ z = \epsilon(q)$.  Then $v, z  \in U_0 \subset {\mathbb C}P^n$.    Recall,  ${\mathcal M }$ be the Hilbert space of holomorphic sections of $H^{\otimes m} \otimes T^{*(n,0)}( {\mathbb C}P^n)$ where $H$ is the hyperplane section bundle and $m$ is a positive integer.  
The vector space $\epsilon^*({\mathcal M})$ is a Hilbert space with the norm defined as 
\begin{equation}\label{norm}
\lvert  \lvert \tilde{s}  \rvert \rvert_{X} = \rm{min}_{s \in {\mathcal M}} \{ \lvert \lvert  s   \rvert \rvert _{{\mathbb C}P^n} :  \tilde{s}= s  \circ \epsilon \}.
\end{equation}
 where $\tilde{s}$,  a typical element of $\epsilon^*({\mathcal M})$,  is a section of the line bundle $\epsilon^*(H^{\otimes m} \otimes T^{* (n,0)} (\mathbb{C}P^n) )$.
 
 We know $\epsilon(X) =  \cup_{\alpha=0}^{n} \epsilon(X) \cap U_{\alpha}$, where $U_{\alpha}$ are the inhomogeneous coordinate neighbourhoods of ${\mathbb C}P^n$.   Let  $W_{\alpha} = \epsilon(X) \cap U_{\alpha}$.  Then $\cup_{\alpha} V_{\alpha}  = \cup_{\alpha} \epsilon^{-1}(W_{\alpha})$,   $\alpha =0, 1,2,...,n$ is an open cover of  $X$. 
Let $v = \epsilon(p)$. 
Define $$Q_{\alpha}(\bar{v}, \cdot) := \epsilon^*K_{\alpha} (\bar{v}, \cdot)$$  locally on $V_{\alpha}   =  \epsilon^{-1}(W_{\alpha})$.   Let
 $W_{\beta} = \epsilon(X) \cap U_{\beta}$ and $V_{\beta} = \epsilon^{-1}(W_{\beta})$. 
 Let $p \in V_{\alpha} \subset X$ and $q \in V_{\beta} \subset X$ .
We can define $Q_{\alpha \beta}$ as $$Q_{\alpha} (\bar{v}, z) = Q_{\alpha \beta} (\bar{v}, z) \tau_{\beta}$$ where $v = \epsilon(p), z = \epsilon(q)$, $\tau_{\beta}   = \epsilon^*(s_{\beta} d z_1 \wedge...\wedge d z_n)    $,  $s_{\beta}$ being  the unit section of $H^{\otimes m}$ on $U_{\beta} \subset {\mathbb C}P^n$ and $z_1,...., z_n$ are coordinates on $U_{\beta}$.
Note $\tau_{\beta}$ is a section of $\epsilon^*(H^{\otimes m} \otimes T^{(n,0)}({\mathbb C}P^n))$.

One can show that $Q_{\alpha \beta} (\bar{v}, z)$ (where $v = \epsilon(p)$ and $z= \epsilon(q)$)   do satisfy the conditions $(2.13)$ and $(2.14)$ in \cite{Od} because it is induced from the pullback of the reproducing kernel  and by \cite{PaRa} (Prop. (5.6) and Th. (5.7)) these pull back kernels are reproducing kernels.  Thus if we 
let $p \in V_{\alpha} \subset X$ and $v = \epsilon(p) \in \epsilon(X) \cap U_{\alpha}$ and 
 $\tilde{s} \in \epsilon^*({\mathcal M})$, i.e.  $\tilde{s} = \epsilon^*(\Psi_{\alpha} d v^1_{\alpha} \wedge ...\wedge d v^n_{\alpha})$ on $V_{\alpha}$,  we have 
 by reproducing kernel property,   $\Psi_\alpha(\epsilon(p)) = \langle Q_{\alpha} (\bar{v},  \cdot) ,  \tilde{s} \rangle_X$.
We have 
\begin{equation}
Q_{\alpha \alpha} (\bar{z}, z) >0,
\end{equation}
\begin{equation}
Q_{\alpha \beta}( \bar{v}, z) = \langle Q_{\beta}(\bar{v},  \cdot),  Q_{\alpha} (\bar{z}, \cdot) \rangle_X,
\end{equation}
where $v = \epsilon(p) \in U_{\alpha} $ and $z = \epsilon(q) \in U_{\beta}$.  The inner product $\langle \cdot, \cdot \rangle_X$ is induced from the norm $\lvert \lvert \cdot \rvert \rvert _X$ defined as in Eq. (\ref{norm}).

Let $V_{\alpha}$ and $V_{\beta}$ be two  open neighbourhoods belonging to the chart on $X$ mentioned above. 
Exactly as in \cite{Od} we can define the transition probability between $p,q$ on $X$ (where $p \in V_{\alpha}$, $q \in V_{\beta}$) as 
$$A_{\alpha \beta}(\bar{v}, z) := \frac{Q_{\alpha \beta}(\bar{v}, z)}{Q_{\alpha \alpha} (\bar{z}, z) ^{\frac{1}{2}} Q_{\beta \beta}(\bar{v}, v)^{\frac{1}{2}}},$$ where $v= \epsilon(p)$ and $z = \epsilon(q)$.

Let us take a path $\gamma$ on $\epsilon(X)$ which joins $z$ and $v$.  We can define $A(\gamma, p, q) $ analogous to $a(\gamma, \bar{v}, z)$ as in \cite{Od}. Then by exactly same argument as in proposition \ref{holonomy} we have
$$A(\gamma, p, q) = \exp [i \int_{\gamma} \rm{Im} (\bar{\partial} \log Q)].$$

\begin{remark}
 Parallelly as in \cite{Od},  this  can be interpreted as  parallel transport  w.r.t.  a certain connection of the pullback  $\epsilon^*  ( H^{\otimes m} \otimes T^{*(n,0)}({\mathbb C}P^n))$ on $X$ and we can define the path integral analogously as an  integral over all paths in $\Sigma$ joining $z$ and $v$ as 
$$ A(p,q) = A_{\alpha \beta}(\bar{v}, z) := \int {\mathcal D} [\gamma] \exp[i \int_{\gamma} \rm{Im} (\bar{\partial} \log Q)].$$
\end{remark}

\section{Acknowledgments}
Rukmini Dey acknowledges support from the
project RTI4001, Department of Atomic Energy, Government of India.
This paper is a part of a presentation  in a  conference,  ``The 33rd/35th International Colloquium on Group Theoretical  Methods in Physics'',  Cotonou,  Benin,  July 15 - 19, 2024; https://icgtmp.sciencesconf.org/.
She  grateful to the organizers for inviting her  to give a talk in this conference.


\begin{thebibliography}{99}




\bibitem{AAGM}  S.T.  Ali,  J.-P.  Antoine  J.  P.  Gazeau and U.A. Mueller,  Coherent States and their Generalizations:a mathematical overview,   {\it Reviews in Mathematical Physics }{\bf 7}(1995) 1013-1104.
 

\bibitem{AAGOS}   A.  Strasburger (ed.),  S.T.  Ali (ed.),  J.-P.  Antoine (ed.) and A.  Odzijewicz (ed.), Quantization,Coherent States and Poisson Structures,  {\it  Proceedings of the XIV Workshop on Geometric Methods in Physics} ( PWN, Warszawa 1998).



\bibitem{Be} F.  A.  Berezin,   Quantization,  {\it Math USSR Izvestija}  {\bf 8}  (1974) 1109-1165 .



\bibitem{BG} D. C.  Dorje,  E.-M. Graefe,   Coherent States and Rational Surfaces,  {\it   J.Phys.A} {\bf 43} (2010) 255205. 

	
\bibitem{DeGh} R.  Dey and K.  Ghosh,
Pull back coherent states and squeezed states and quantization,  {\it Symmetry,  Integrability  and Geometry: Methods and Applications (SIGMA)}{\bf 18},   028, (2022) 1-14,   arxiv: 2108.08082.



	
\bibitem{DeGh2} R.  Dey and K.  Ghosh, Berezin-type quantization on even-dimensional compact manifolds,  {\it J.   Phys.: Conf. Ser. } {\bf 2667} (2023) 012003.

 
 
 
 
 \bibitem{E}
M.   Englis,  Berezin Quantization and Reproducing Kernel on Complex Domains {\it  Trans. Amer. Math. Soc. }
 {\bf 348} (1996)411. 

 \bibitem{Fe} B.V. Fedosov,  A Simple Geometrical Construction of Deformation Quantization,  {\it J . Diff.  Geom}, {\bf 40 }(1994)213-238.
 
 
 \bibitem{GM} J.  P.  Gazeau and P.  Moncea, Generalized Coherent States for Arbitrary Quantum Systems,{\it
 Mathematical Physics Studies book series} { \bf 21/22} (MPST,  Springer, 2000). 
 
 
 
 \bibitem{GoSh} G. A. Goldin and D. H. Sharp, A Universal Kinematical Group for Quantum Mechanics,  arxiv: 2404.18274 (2024).
 
 \bibitem{Gh} K.  Ghosh, Berezin-type quantization of even-dimensional manifolds and pullback coherent states,  PhD. Thesis,  ICTS-TIFR,  https://thesis.icts.res.in/ (2023).
	
 \bibitem{HCPL} W. Ho,  S. Choi,  H.  Pichler,  M.  Lukin, Periodic orbits, entanglement and quantum many body scars in constrainted models: matrix product state approach,  {\it Phys. Rev. Lett.} {\bf 122} (2019) 040603.
 
 \bibitem{Ka} J. R. Klauder,  Quantization is Geometry,  After All,  {\it  Ann. Phys. } {\bf 188} (1988) 120-130.
 
 \bibitem{Ka2} J. R. Klauder and E. Onofri,   Landau Levels and Geometric Quantization,  {\it Int.  J.  Mod.  Phys.  A} {\bf  15} (1989) 3939-3949.
 
 \bibitem{KarNa} D. Karabali and V.P. Nair, The effective action for edge states in higher dimensional quantum Hall systems,  {\it Nucl.Phys. B}  {\bf 679} (2004) 427-446. 
 
 
 \bibitem{Kar}D. Karabali,  Entanglement entropy for integer quantum Hall effect in two and higher dimensions, 
{\it Phys. Rev. D} {\bf 102}(2020) 025016. 
  
  
  \bibitem{Ko} S. Kolodziej, The Monge-Ampere Equation on Compact Kahler manifolds, {\it Ind. Univ. Math. J.} {\bf 52} 3 (2003) 667-686.  
  
 \bibitem{Od} A.  Odzijewicz,  On Reproducing Kernels and Quantization of States, {\it  Commun. Math. Phys.}  {\bf 114},    (1988) 577--597.
 
 \bibitem{Od2} A. Odzijewicz,  Coherent States and Geometric Quantization,  {\it Commun. Math. Phys. } {\bf 150 }(1992) 577-576.
 
 
 \bibitem{OMY} H. Omori, Y. Maeda, A. Yoshioka,  Weyl manifolds and defomation quantization, Adv. in Math. 85 (1991) , 224-255. 

 \bibitem{PaRa} V.  I.  Paulsen and M.  Raghupathi,  {\it An Introduction to the Theory of Reproducing Kernel Hilbert Spaces}  Cambridge Studies in Advanced Mathematics {\bf 152} (Cambridge University Press 2016).
 
\bibitem{Pe} A. Perelomov,  {\it Generalized coherent states and their applications,} Texts and Monographs in Physics (Springer-Verlag 1986).


\bibitem{Rad} J.  M.  Radcliffe, Some Properties of Coherent Spin States, {\it J. Phys. A: Gen. Phys.} {\bf 4} (1971) 313.
	

\bibitem{Ra}  J.  H.  Rawnsley,  Coherent States and Kahler Manifolds,  {\it Quart. J. Math.} {\bf 28} (1977) 403-415 .
	

\bibitem{Sp} M.  Spera, 
 On K$\ddot{\rm{a}}$hlerian Coherent States,   {\it Proceedings of the International Conference on Geometry, Integrability and Quantization} {\bf 241-256}  (Bulgarian Academy of Sciences,  2004).

 \bibitem{To} D. Tong, The Quantum Hall Effect,  {\it Tata Infosys Lectures},
 
   http://www.damtp.cam.ac.uk/user/tong/qhe.html (2016).



\end{thebibliography}
\end{document}